\documentclass[a4paper, notitlepage, 11pt]{article}
\usepackage[utf8]{inputenc}
\usepackage[english]{babel}
\usepackage[T1]{fontenc}
\usepackage{amsmath}

\usepackage{amsthm,amssymb,amscd, physics}
\usepackage{color}
\usepackage{savesym}
\usepackage{soul}
\usepackage{url}
\usepackage[normalem]{ulem}
\usepackage{enumitem}
\usepackage[left=2.5cm,right=2.5cm,bottom=2.5cm,top=2.5cm]{geometry}
\usepackage{tikz}
\usepackage{appendix}
\usepackage{csquotes}
\usepackage{float}
\usepackage{fancyhdr}
\usepackage{mathtools}
\usepackage{stmaryrd}
\usepackage{thm-restate}
\usepackage{authblk}
\usepackage[hidelinks]{hyperref}
\usepackage[sort&compress,square,numbers]{natbib}
\usepackage{diagbox}
\usepackage{bbm}
\savesymbol{div}
\usepackage{mathabx}
\restoresymbol{mathabx}{div}

\numberwithin{equation}{section}

\newcommand{\cket}[1]{|#1\rangle}
\newcommand{\bracket}[2]{\langle #1 | #2\rangle}
\newcommand{\mcl}[1]{\ensuremath{\mathcal{#1}}}

\newcommand{\Acal}{{\mcl A}}
\newcommand{\Bcal}{{\mcl B}}

\newcommand{\EcalKDC}{{\mcl E}_{\mathrm{KD+}}}
\newcommand{\EcalKDCpu}{{\mcl E}_{\mathrm{KD+}}^{\mathrm{pure}}}

\newcommand{\spanR}{{\mathrm{span}_{\R}}}
\newcommand{\Hcal}{{\mcl H}}

\newcommand{\N}{\mathbb{N}}
\newcommand{\Z}{\mathbb{Z}}
\newcommand{\R}{\mathbb{R}}

\newcommand{\C}{\ensuremath{\mathbb{C}}}

\newcommand{\bbone}{\mathbb{I}}

\newcommand{\conv}[1]{\mathrm{conv}\left(#1\right)}

\newcommand{\card}[1]{\left|#1\right|}

\newcommand{\VR}{V_{\mathrm{KDr}}}

\renewcommand{\H}{\mcl{H}}

\newcommand{\IntEnt}[2]{\llbracket #1 , #2 \rrbracket}
\newcommand{\convAB}{\conv{\Acal \cup \Bcal}}
\renewcommand{\epsilon}{\varepsilon}

\newtheorem{Theorem}{Theorem}[section]

\newtheorem{Cor}[Theorem]{Corollary}
\newtheorem{definition}[Theorem]{Definition}

\newtheorem{Lemma}[Theorem]{Lemma}

\begin{document}

\title{The Kirkwood-Dirac representation associated to the Fourier transform for finite abelian groups: positivity}

\author[,1]{Stephan De Bi\`evre\thanks{stephan.de-bievre@univ-lille.fr}}

\author[,1]{Christopher Langrenez\thanks{christopher.langrenez@univ-lille.fr}}

\author[,2]{Danylo Radchenko\thanks{danylo.radchenko@univ-lille.fr}}

\affil[1]{Univ. Lille, CNRS, Inria, UMR 8524, Laboratoire Paul Painlev\'e, F-59000 Lille, France}

\affil[2]{Univ. Lille, CNRS, Laboratoire Paul Painlev\'e, F-59000 Lille, France}

\maketitle

\begin{abstract}
We construct and study the Kirkwood-Dirac (KD) representations naturally associated to the Fourier transform of finite abelian groups $G$. We identify all pure KD-positive states and all KD-real observables for these KD representations. We provide a necessary and sufficient condition ensuring  that all KD-positive states are  convex combinations of pure KD-positive states. We prove  that for $G=\Z_{d}$, with $d$ a prime power, this condition is satisfied.  We provide examples of abelian groups where it is not. In those cases, the convex set of KD-positive states contains states outside the convex hull of the pure KD-positive states. 
\end{abstract}

\section{Introduction}
Quasiprobability representations have played an important role in the development of quantum mechanics. The best known --- and the oldest --- of these representations are the ones based on the use of the Wigner function~\cite{wigner1932} (or Weyl symbol), as well as the related Glauber-Sudarshan $P$ function and the Husimi function~\cite{Husimi40,Sudarshan63,Glauber63}, which are all intimately linked to the presence of conjugate variables in the theory, that obey the canonical commutation relations. Hence, such quasiprobability representations are closely related to the Fourier transform on $\R^n$ (or on $\Z_{d}$) and to the irreducible representations of the Heisenberg group of $\R^n$ (or of $\Z_{d}$). They provide phase-space representations of quantum mechanics in which quantum states are represented by quasiprobabilities, and observables by real-valued functions on phase space. These quasiprobability representations have found numerous applications in the study of quantum foundations, in quantum optics, and in quantum information theory. We refer to~\cite{leonhardt,Serafini,hudson1974wigner,cohen1966,cagl69a,cagl69b,Fe11,spekkens2008} and references therein for further details on such applications.
More recently, the Kirkwood-Dirac quasiprobability representations of quantum mechanics have come to the forefront in various parts of quantum physics. They provide an increased flexibility and range of applicability, since they are not contingent on the presence of conjugate variables in the theory. For extensive recent reviews of the research on Kirkwood-Dirac distributions, and of their applications to physics, we refer to~\cite{arvidssonshukur2024properties,lostaglio2023kirkwood}. In this paper, we construct and study the Kirkwood-Dirac representation naturally associated to the Fourier transform of a finite abelian group $G$. 

A concise definition of a general Kirkwood-Dirac representation can be given as follows.  Let $\Hcal$ be a $d$-dimensional Hilbert space and let $A$ and $B$ be two self-adjoint operators on $\Hcal$, with respective eigenbases $\left\{\cket{a_i}\right\}_{i}$ and  $\left\{\cket{b_j}\right\}_{j}$, and eigenvalues $\left(a_i\right)_{i}$ and $\left(b_j\right)_{j}$. We will always assume that the transmission matrix $U_{ij}=\langle a_i|b_j\rangle$ does not have zeroes and write $m=\min_{ij}|\langle a_i|b_j\rangle|$. Introducing
\begin{equation}\label{eq:frames}
    S_{ij}=\bracket{a_i}{b_j} \cket{a_i}\bra{b_j},\quad \widetilde{S}_{ij}=|\langle a_i|b_j\rangle|^{-2}S_{ij},
\end{equation}
we define the lower, respectively upper, Kirkwood-Dirac symbol of an operator $C$ by 
\begin{equation}
    Q_{ij}[C]= \Tr (S_{ij}^{\dagger} C),\quad \widetilde{Q}_{ij}[C]= \Tr (\widetilde{S}_{ij}^{\dagger} C).
\end{equation}
The pair $Q, \widetilde{Q}$ is said to be a Kirkwood-Dirac quasiprobability representation of quantum mechanics on $\Hcal$~\cite{Fe11}. Although this representation depends strongly on the choice of $A$ and $B$, we shall not indicate this choice in the notation. The choice of $A$ and $B$  must be adapted to the physical system considered which  implies that the Kirkwood-Dirac representation can be adapted to many different systems, as indicated above. One has
\begin{equation}\label{eq:quasiprob}
   \Tr(C)=\sum_{ij} Q_{ij}[C], \quad \langle a_i|C|a_i\rangle=\sum_j Q_{ij}[C],\quad \langle b_j|C|b_j\rangle=\sum_i Q_{ij}[C],
\end{equation}
and
\begin{equation}\label{eq:overlap}
   \Tr D^\dagger C=\sum_{ij} \overline{\widetilde{Q}_{ij}[D]}Q_{ij}[C].
\end{equation}
This equation is often referred to as the overlap identity. The map $C\to Q[C]$ is a bijection and one can reconstruct $C$ from its Kirkwood-Dirac symbol:
\begin{equation}\label{eq:reconstruction}
   C = \sum_{(i,j)\in\IntEnt{1}{d}} Q_{ij}[C]\frac{\cket{a_i}\bra{b_j}}{\bracket{b_j}{a_i}}.
\end{equation}
When $C=\rho$ is  a density matrix, meaning a positive operator of unit trace, Eq.\eqref{eq:quasiprob} shows $Q_{ij}[\rho]$ has the properties of a joint probability measure with marginals the Born rule probabilities for the observables $A$ and $B$, except that $Q_{ij}[\rho]$ is complex valued, not necessarily positive. This explains why $Q_{ij}[\rho]$ is referred to as the Kirkwood-Dirac quasiprobability distribution of $\rho$. It is well-known that it is not possible to assign a joint probability distribution  for two non-commuting observables to all density matrices $\rho$; we refer to~\cite{Fe11,spekkens2008,lostaglio2023kirkwood} for proofs of this fact.
Nevertheless, there do always exist some density matrices $\rho$ for which $Q_{ij}[\rho]\geq 0$ for all $i,j$: in that case $Q_{ij}[\rho]$ may be thought of as a joint probability distribution.  Such density matrices are said to be Kirkwood-Dirac positive. We will denote the convex set of Kirkwood-Dirac-positive density matrices by $\EcalKDC$. We will further denote  by $\VR$ the real vector space of self-adjoint operators that have a real upper Kirkwood-Dirac symbol. Note that, since $\widetilde{Q}_{ij}(\rho) = \frac{1}{\left|\bracket{a_i}{b_j}\right|^{2}}Q_{ij}(\rho)$ for any density matrix $\rho$, the set of Kirkwood-Dirac-positive states $\EcalKDC$ is included in $\VR$. 
In quantum mechanics, states are represented by density matrices $\rho$ and observables by self-adjoint operators $D=D^\dagger$, with $\Tr(D\rho)$ representing the expectation value of $D$ in the state $\rho$. The ``overlap formula'' Eq.\eqref{eq:overlap} allows us to interpret $\Tr(D\rho)$ as the expectation value of the random variable $\widetilde{Q}_{ij}[D]$ with respect to the quasiprobability $Q_{ij}[\rho]$. For Kirkwood-Dirac-positive density matrices $\rho\in\EcalKDC$ and for observables $D\in \VR$, this becomes: 
\begin{equation}
    \Tr D\rho= \sum_{ij} \widetilde{Q}_{ij}[D]Q_{ij}[\rho],
\end{equation}
which can now be understood as a  bona fide expectation value of the  real random variable $\widetilde{Q}_{ij}(D)$ with respect to the probability measure $Q_{ij}(\rho)$. In this manner one obtains a true probability representation of the ``fraction'' of quantum mechanics given by the states in $\EcalKDC$ and the observables in $\VR$. We further point out that it has been shown in the context of quantum metrology \cite{arvidsson-shukur2020} and quantum simulation \cite{pashayan2015} that, to obtain a quantum advantage, one needs Kirkwood-Dirac-nonpositivity of $\rho$: in this sense, the positive fraction of quantum mechanics referred to above can be thought of as classical; see Ref.~\cite{arvidssonshukur2024properties} for details on these issues.

Let us  point out that similar questions arise in quantum optics, concerning the positivity of the Wigner and Glauber-Sudarshan quasiprobability distributions: in each case, the full identification or characterization of the set of states with a positive quasiprobability distribution is of considerable current interest; we refer to~\cite{gross2006,vanherstraetencerf2021,herstraetenetal2023,diasprata2023a,Bievre19,ryspagal15}  and references therein for further details. 

The above considerations naturally lead to the question of how to  characterize $\EcalKDC$.
It is easy to check that 
\begin{equation}
    \convAB\subseteq \EcalKDC,
\end{equation}
where
\begin{equation}
    \Acal=\{|a_i\rangle\langle a_i||1\leq i\leq d\},\quad \Bcal=\{|b_j\rangle\langle b_j| |1\leq j\leq d\}.
\end{equation}
Denoting by $\EcalKDCpu$ the set of pure Kirkwood-Dirac-positive states, and recalling that a pure state is a rank one orthogonal projector, one therefore has
\begin{equation}\label{eq:KDgeom}
    \convAB\subseteq \conv{\EcalKDCpu}\subseteq \EcalKDC.
\end{equation}
For an arbitrary choice of $\Hcal, A$ and $B$, it is not straightforward to describe the sets $\EcalKDCpu$ and $\EcalKDC$ explicitly and in particular to decide which of the two inclusions in Eq.\eqref{eq:KDgeom} are equalities, if any.  Results on these questions have been obtained recently in a variety of situations. 

In \cite{arvidsson-shukuretal2021,debievre2021,DeB23}, a general link is established between Kirkwood-Dirac positivity of pure states and uncertainty: pure Kirkwood-Dirac-positive states have low uncertainty. In
\cite{DeB23}, it is further shown that, provided $d$ is a prime number and the transition matrix $U$ between the two bases is  the discrete Fourier transform, the only pure Kirkwood-Dirac-positive states are the basis states. In other words, in that case, $\EcalKDCpu=\Acal\cup\Bcal$. In~\cite{XU2024}, Xu extended this result, provided the dimension $d$ is prime,   to general mutually unbiased bases (MUB bases), for which $|\langle a_i|b_j\rangle|=d^{-1/2}$.
Concerning mixed states, it was proven in \cite{langrenezetal2024}  that, if the bases $(\cket{a_i})_{i\in\IntEnt{1}{d}}$, $(\cket{b_j})_{j\in\IntEnt{1}{d}}$ are chosen uniformly with respect to Haar measure, then, with probability one, all inclusions in Eq.\eqref{eq:KDgeom} are equalities, so that
\begin{equation}\label{eq:simple}
    \convAB= \conv{\EcalKDCpu} = \EcalKDC.
\end{equation}
In addition, in those cases, the set of observables with real Kirkwood-Dirac symbol is explicitly identified as
$$
\VR=\spanR{\{|a_i\rangle\langle a_i|, |b_j\rangle\langle b_j|\mid 1\leq i,j\leq d\}}.
$$
Examples where the inclusions in Eq.\eqref{eq:KDgeom} are strict, so that the geometry of $\EcalKDC$ and of $\VR$ is considerably more complex are given in~\cite{langrenez2023characterizing2}. For example, $\EcalKDC$ is then not necessarily a polytope and its extreme points can be hard to identify. In such cases, witnesses of nonpositivity can be used to study the geometry of Kirkwood-Dirac-positive states.
In \cite{langrenezetal2024a}, 
 such witnesses are constructed and the link between Kirkwood-Dirac positivity and uncertainty for pure states that were shown to exist~\cite{DeB23}, is extended to mixed states. 
 
 These results indicate that, given two arbitrary bases $(|a_i\rangle)_i, (|b_j\rangle)_j$, it is not straightforward to determine the basic features of $\EcalKDCpu$ and of $\EcalKDC$. When extra mathematical structures are present, this task can be expected to be more accessible.  For example,  in any dimension, the Kirkwood-Dirac positive pure states associated to the discrete Fourier transform form a finite family that is easily described~\cite{XU2024,MaOzPr04}. 
 It is in addition shown in~\cite{langrenez2023characterizing2} that, in that case, and provided the dimension is a prime number, Eq.\eqref{eq:simple} does hold.  When the dimension  $d=p^{2}$ with $p$ a prime number, it is proven in~\cite{Yang_2024} that 
 \begin{equation}\label{eq:almostsimple}
    \convAB\subsetneq \conv{\EcalKDCpu} = \EcalKDC.
\end{equation}
In this situation, as mentioned above, the pure Kirkwood-Dirac positive states are explicitly known, and, as in Eq.\eqref{eq:simple}, the set of Kirkwood-Dirac-positive states is a polytope.

 In this article, we consider the family of Kirkwood-Dirac quasiprobability representations that are naturally associated to the Fourier transform on finite abelian groups $G$, as follows. We denote by $\widehat{G}$ the unitary (or Pontyagrin) dual of $G$ and by $\H$
the Hilbert space $L^{2}(G)$, on which we consider the two natural orthonormal bases given by
\begin{equation}\label{eq:BasisStates}
 \forall g'\in G, a_{g}(g') =  \delta_{g}(g') \text{ and } b_{\chi}(g') = \frac{1}{\left|G\right|^{\frac{1}{2}}}\chi(g')
\end{equation}
where $g\in G$ and $\chi\in\widehat{G}$. We then define the Kirkwood-Dirac representation of the abelian group $G$ to be the one associated to these two bases. In general, the lower and upper symbols $Q$ and $\widetilde{Q}$ are defined on $G\times\hat G$ and they intertwine the left action of $G\times \hat G$ with the natural irreducible representation of the Heisenberg group of $G$ on $\Hcal=L^2(G)$ (Section~\ref{sec:KDfinitegroups}).

When $G=\Z_{d}$, for example, the transition matrix between the above bases is the discrete Fourier transform matrix; the corresponding representation of the Heisenberg group plays an important role in quantum information theory and in quantum chaos, notably.  The representations associated to the group  $\Z_d^N$ arise in the study of systems of $N$ qubits.   

We identify for these Kirkwood-Dirac representations the complete set $\EcalKDCpu$ of pure Kirkwood-Dirac-positive states, which we show to be labeled by the elements of the subgroups of $G\times \hat G$ of the form $H\times H^{\perp}$, where $ H^\perp$ is the annihilator of $H$ (Section~\ref{sec:charKDpurepos}). We also characterize the real vector space of  Kirkwood-Dirac-real observables $\VR$ (Section~\ref{sec:IdKDreal}).   We further derive a necessary and sufficient condition that allows one to check whether $\EcalKDC = \conv{\EcalKDCpu}$ (Theorem~\ref{thm:KDpositivestates} in Section~\ref{sec:CharKDpos}).  Applying these results, we then show (Theorem~\ref{thm:primepowers}) that, when $G=\Z_{d}$ and $d$ is any prime power ($d=p^k, k > 1$), Eq.\eqref{eq:almostsimple} holds, thereby generalizing the results of~\cite{langrenez2023characterizing2} for $d=p$  and~\cite{Yang_2024} for $d=p^2$ to arbitrary prime power dimensions.

Finally, in Section~\ref{sec:ExamplesFT} we give two examples of abelian groups ($G=\Z_6$ and $\Z_2\times \Z_2$) for which Eq.\eqref{eq:almostsimple} does not hold, thus disproving conjectures in~\cite{Yang_2024,xu2024a}. There then exist KD-positive density matrices $\rho$ that have the property that they cannot be written as convex mixtures of pure KD-positive states. Such states have been used in~\cite{thio2024} to construct a fully classical experiment that certifies the existence of experiments that do not admit classical explanations, by which is meant that they cannot be modeled by a noncontextual hidden variable model. 

\section{The Kirkwood-Dirac representation of quantum mechanics over finite abelian groups}
\label{sec:KDfinitegroups}
In this section,  we introduce the Kirkwood-Dirac (KD) representation associated to the Fourier transform on a finite abelian group. 

\subsection{Fourier transform and Heisenberg group}\label{sec:Settings}
We recall some basic facts about the Fourier transform on a finite abelian group and about the representation of the associated Heisenberg group.

Let $(G,+)$ be a finite abelian group, and denote by $\hat G$ its unitary (or Pontryagin) dual. We write $\Hcal=L^2(G)$ and $\widehat\Hcal=L^2(\widehat G)$ and introduce the Fourier transform by
\begin{equation}
\widehat\ \ \colon \psi\in\Hcal\mapsto\widehat\psi\in\widehat\Hcal
\end{equation}
with
\begin{equation}
\widehat\psi(\chi)=\frac1{|G|^{\frac{1}{2}}}\sum_{g\in G}\psi(g)\overline{\chi(g)}.
\end{equation}
The inverse Fourier transform is then given by
\begin{equation}
\widecheck\ \ \colon \eta\in \widehat\Hcal\mapsto \widecheck \eta\in \Hcal
\end{equation}
with
\begin{equation}
\widecheck \eta(g)=\frac{1}{\left|G\right|^{\frac{1}{2}}}\sum_{\chi\in \hat{G}}\chi(g)\eta(\chi).
\end{equation}
Thus $\Hcal$ has two natural orthonormal bases given by $(a_{g})_{g\in G}$ and $(b_{\chi})_{\chi\in \widehat{G}}$ defined in Eq.\eqref{eq:BasisStates}.
Here, the scalar product on $\Hcal$ is
$$
\bracket{\varphi}{\psi} = \sum_{g\in G} \overline{\varphi(g)}\psi(g)=\sum_{g\in G} \bracket{\varphi}{a_g}\bracket{a_g}{ \psi},
$$
and similarly on $\widehat\Hcal$. Note that, with this notation,
\begin{equation}
\psi(g)=\bracket{a_g}{\psi},\quad \widehat\psi(\chi)=\bracket{b_\chi}{\psi}.
\end{equation}
We will denote by $\norm{\psi} = \left(\bracket{\psi}{\psi}\right)^{1/2}$ the associated norm. One may note that the transition matrix between the two bases above is a complex Hadamard matrix:
\begin{equation}
\bracket{a_g}{b_\chi} =\frac1{|G|^{1/2}}\chi(g).
\end{equation}
Such bases are said to be mutually unbiased as they satisfy for all $(g,\chi)\in G\times \hat{G}, \left|\bracket{a_g}{b_\chi}\right| = \card{G}^{-\frac{1}{2}}$. They play an important role in various aspects of quantum information theory.

We also recall some basic facts about the Heisenberg group $(H(G),\cdot)$ and its irreducible representation on $\Hcal$. One defines
\begin{equation}
    H(G)=G\times \hat{G}\times \mcl{S}^1,
\end{equation}
where $\mcl{S}^{1} = \left\{ z\in \C, \left|z\right|=1\right\}$, with the  group operation
\begin{equation}
    (g,\chi,z)\cdot (g',\chi',z') = (g+g',\chi\chi', z z'\overline{\chi'(g)}).
\end{equation}
For each $g\in G, \chi\in \hat G$, we introduce, for all $\psi\in\Hcal$
\begin{equation}
T_{g}\psi(x) = \psi(x-g) \;\text{ and }\;
 M_{\chi}\psi(x) = \chi(x)\psi(x).
\end{equation}
Note, for later purposes, that 
\begin{equation}
    M_\chi T_g=\chi(g)T_gM_\chi.
\end{equation}
We also point out that
\begin{equation}
T_g b_\chi=\overline{\chi(g)}b_\chi \mathrm{\ and \ } M_\chi a_g=\chi(g)a_g.
\end{equation}
In other words, the ``direct'' basis $a_g$ diagonalises the multiplication operators $M_\chi$ whereas the ``dual'' basis $b_\chi$ diagonalises the translation operators $T_g$. 
A unitary irreducible representation of the Heisenberg group $(H(G),\cdot)$ on $L^{2}(G)$ is given by
\begin{equation}\label{eq:Rep1}
U:(g,\chi,z)\in H(G)\to U(g,\chi, z)=zM_{\chi}T_{g}\in U(\Hcal),
\end{equation}
where $U(\Hcal)$ is the unitary group on $\Hcal$.

Note that the group $(H(G),\star)$, with the group operation
\begin{equation}
    (g,\chi,z) \star (g',\chi',z') = (g+g',\chi\chi', z z'\chi(g'))
\end{equation}
is isomorphic to $(H(G),\cdot)$ via the isomorphism
\begin{equation}
    I: (g,\chi,z)\in (H(G), \cdot)  \to   (g,\chi,z\chi(g))\in (H(G), \star).
\end{equation}
As a result, 
\[
U_{\star}:(g,\chi,z)\in H(G)\to zT_{g}M_{\chi}\in U(\Hcal),
\]
defines a unitary irreducible representation of $(H(G),\star)$, since 
$$
U_\star(g,\chi,z)=U(I^{-1}(g,\chi,z)).
$$
For the remainder of the paper, we will use Heisenberg group in the form $(H(G),\cdot)$, with the associated representation $U$ given in Eq.\eqref{eq:Rep1}. This corresponds to a choice of ordering, as can be seen in Eq.\eqref{eq:ordering2} and Eq.\eqref{eq:WHactKD} below:  ``multiplication operators to the left of translation operators''. Using $(H(G),\star)$ and $U_{\star}$, one obtains the opposite ordering. The results of this paper go through unaltered with that choice.  

\subsection{The Kirkwood-Dirac representation for a finite abelian group}\label{sec:KDsymbol}
We now define the KD representation associated to the Fourier transform of a finite abelian group, following Eq.\eqref{eq:frames}-\eqref{eq:reconstruction}. For that purpose, we need to introduce a frame $S$ and its dual frame $\widetilde{S}$ on $L^2(G)$. We proceed as follows.

We define
\begin{eqnarray}\label{eq:frame} K_\chi=|b_\chi\rangle\langle b_\chi|,\quad
 L_g=|a_g\rangle\langle a_g|, \mathrm{ \ and \ }  S(g, \chi)= L_g K_\chi.
\end{eqnarray}
Note that the $S(g,\chi)$ form an orthogonal basis of $\mathcal L(\mathcal H)$ as
\begin{equation}
    \Tr S(g,\chi)^\dagger S(g',\chi')=|G|^{-1}\delta_{g,g'}\delta_{\chi,\chi'},
\end{equation}
and that the $\widetilde{S}(g,\chi)=|G|S(g, \chi)$ form its dual basis. 
\begin{definition}
The lower Kirkwood-Dirac symbol of an operator $C\in\mathcal L(\H)$   is the function $Q[C]$ on $G\times \hat G$ defined by
\begin{equation}
Q[C](g,\chi)=\bracket{b_\chi}{a_g}\bra{a_g}C\cket{ b_\chi}=\Tr(S(g,\chi)^\dagger C). 
\end{equation}

The upper Kirkwood-Dirac symbol  of an operator $C\in\mathcal L(\H)$ is the function $\widetilde{Q}[C]$ on $G\times \hat G$ defined by 
\begin{equation}
\widetilde{Q}[C](g,\chi)=\frac{\bra{a_g}C \cket{b_\chi}}{\bracket{a_g} {b_\chi}}=\Tr(\widetilde{S}(g,\chi)^\dagger C). 
\end{equation}
\end{definition}
Note that the maps
\begin{equation}\label{eq:Qmapdef}
Q: C\in \mathcal L(\Hcal)\to Q[C] \in L^2(G\times\widehat G),\quad
\widetilde{Q}: C\in \mathcal L(\Hcal)\to \widetilde{Q}[C] \in L^2(G\times\widehat G) 
\end{equation}
are linear.
One finds that
\begin{equation}
  \widetilde{Q}[\bbone_{\mathcal H}](g,\chi)=1,\quad   \Tr C=\sum_{g,\chi} Q[C](g,\chi),
\end{equation}
and, for the special case of a rank $1$ operator denoted by $C = \cket{\psi}\bra{\psi} $, we have:
\begin{equation}\label{eq:KDsymbPure}
    Q[\psi](g,\chi) = \frac{1}{\card{G}^{\frac{1}{2}}}\overline{\chi(g)}\psi(g)\overline{\widehat{\psi}(\chi)} \mathrm{\ and \ } \widetilde{Q}[\psi](g,\chi) = \card{G}^{\frac{1}{2}}\overline{\chi(g)}\psi(g)\overline{\widehat{\psi}(\chi)},
\end{equation}
where $Q[\psi]$ is a shorthand for $Q[\cket{\psi}\bra{\psi}]$.
The inverses of $Q$ and $\widetilde{Q}$ are readily computed. Given $f\in L^2(G\times \hat G)$, one has
\begin{equation}
\widetilde{Q}^{-1}[f]=\sum_{g,\chi} \bracket{a_g}{b_\chi} f(g,\chi) \cket{a_g}\bra{b_\chi},
\end{equation}
and
\begin{equation}
Q^{-1}[f]=\card{G} \sum_{g,\chi} \bracket{a_g}{b_\chi} f(g,\chi) \cket{a_g}\bra{b_\chi}.
\end{equation}
One further readily checks that Eq.\eqref{eq:overlap} is satisfied: for all $C, C'\in\mathcal L(\Hcal)$,
\begin{equation}\label{eq:Overlap}
\Tr C^{\dagger}C'=\sum_{g, \chi} \overline{\widetilde{Q}[C](g,\chi)} Q[C'](g,\chi).
\end{equation}

It is instructive to consider some special cases of the correspondence between operators and their KD symbols. First, let $v:G\to\C$ and $w:\widehat G\to \C$, then, if 
\begin{equation}
V=\sum_g v(g) |a_g\rangle\langle a_g|,\quad W=\sum_{\chi} w(\chi)|b_\chi\rangle\langle b_\chi|,
\end{equation}
then
\begin{equation}\label{eq:ordering1}
   \widetilde{Q}[VW](g,\chi)=v(g)w(\chi)=w(\chi)v(g),
\end{equation}
so that 
$\widetilde{Q}[V]=v$, $\widetilde{Q}[W]=w$, and 
\begin{equation}\label{eq:ordering2}
\widetilde{Q}[VW]=\widetilde{Q}[V]\widetilde{Q}[W]=\widetilde{Q}[W]\widetilde{Q}[V]\not=\widetilde{Q}[WV].
\end{equation}
This corresponds to ``left ordering'': operators diagonal in the direct basis have symbols that are functions depending on $g\in G$ only and go to the left of operators diagonal in the dual basis, and that have symbols depending only on the dual variable $\chi\in\hat G$. 

We also have that
\begin{equation}
    \widetilde{Q}[M_{\chi'}](g,\chi)=\chi'(g), \quad \widetilde{Q}[T_{g'}](g,\chi)=\overline{\chi(g')}.
\end{equation}

Finally, we note that the Heisenberg group acts on the set of self-adjoint operators. We want to know how this action acts on $Q$. Let $F$ be a self-adjoint operator, for $(z_{0},g_0,\chi_0)\in \mcl{S}^1\times G\times\hat{G}$, we have that:
\begin{eqnarray}\label{eq:WHactKD}
    \forall (g,\chi)\in G\times\hat{G}, Q[U(g_0,\chi_0,z_0)F U^{\dagger}(g_0,\chi_0,z_0)](g,\chi) = Q[F](g-g_0,\chi\overline{\chi_{0}}).
\end{eqnarray}
For later purposes, we note that Eq.\eqref{eq:WHactKD} implies that the action of the Heisenberg group on self-adjoint operators induces  translations of $Q$ by elements of $G\times\hat G$. 

Our interest is to characterise the set of quantum states that have a positive KD distribution, which means that the KD distribution $Q[\rho]$ of a state $\rho$ has positive values everywhere. In the next section, we identify the set of pure KD-positive states, which fully characterises the set $\conv{\EcalKDCpu}$.

\section{Characterisation of the Kirkwood-Dirac-positive pure states for finite abelian groups}
\label{sec:charKDpurepos}
The aim of this section is to characterise the set of KD-positive pure states. This will be done in Theorem~\ref{thm:KDPure}. For that purpose, we first identify a particular set of KD-positive states and then prove that any KD-positive pure state is of this specific form.

Let $H\subseteq G$ be a subgroup of~$G$.  We define the normalized characteristic function of $H$ as 
\begin{equation}\label{def:KDPure1}
\psi^{H} = \frac{1}{\left|H\right|^{\frac{1}{2}}}\mathbbm{1}_{H}.
\end{equation}
Since
\begin{equation}
    \widehat{\psi^H}=\frac{1}{\left|H^\perp\right|^{\frac{1}{2}}}\mathbbm{1}_{H^\perp},
\end{equation}
and $|G|=|H||H^\perp|$,
the
 KD distribution associated to this state is given by
\[
 \forall (g,\chi)\in G\times\hat{G}, \ Q\left[\psi^{H}\right](g,\chi)= \frac{1}{\left|G\right|}\overline{\chi(g)}\mathbbm{1}_{H}(g)\mathbbm{1}_{H^{\perp}}(\chi) = \frac{1}{\left|G\right|}\mathbbm{1}_{H}(g)\mathbbm{1}_{H^{\perp}}(\chi),
\]
where $H^{\perp}$ is the annihilator of $H$ defined as $H^{\perp} = \left\{ \chi \in \hat{G}, \forall h \in H, \chi(h) = 1 \right\}$. Thus, $\psi^{H}$ is a KD-positive state.

Furthermore, according to Eq.\eqref{eq:WHactKD}, the action of the Heisenberg group on a KD-positive pure state preserves KD-positivity. Thus, for any $(g_{0},\chi_{0})\in G/H \times \hat{G}/H^{\perp}$, we have that 
\begin{equation}\label{def:KDPure2}
    \psi^{H}_{g_{0},\chi_{0}}:= M_{\chi_0}T_{g_{0}}\psi^{H}
\end{equation} 
is a pure KD-positive state. Note that if $H$ is reduced to $\{0\}$, then we obtain the pure states $(a_g)_{g\in G}$ and if $H=G$, we obtain the pure states $(b_{\chi})_{\chi\in\hat{G}}$. The KD distributions of the $\psi^{H}_{g_{0},\chi_{0}}$ are:
\begin{equation}\label{eq:KDPosPure}
 \forall (g,\chi)\in G\times\hat{G}, \eta^{H}_{g_0,\chi_0}(g,\chi): =Q\left[\psi^{H}_{g_{0},\chi_{0}}\right] (g,\chi)= \frac{1}{\left|G\right|}\mathbbm{1}_{H}(g-g_0)\mathbbm{1}_{H^{\perp}}(\chi\chi_{0}^{-1}).
\end{equation}

The following Theorem ensures that any  KD-positive pure state is of the form given in Eq.\eqref{eq:KDPosPure}.
\begin{Theorem}\label{thm:KDPure}
     Let $G$ be a finite abelian group. Then $\psi$ is a pure KD-positive state if and only if there exists a subgroup $H$ and $(g_0,\chi_0)\in G/H \times\hat{G}/H^{\perp}$ such that $\psi = \psi^{H}_{g_0,\chi_0}$.
 \end{Theorem}

\begin{proof}
    As $\psi^{H}_{g_0,\chi_0}$ is a KD-positive pure state for all subgroups $H$ and all $(g_0,\chi_0)\in G/H \times\hat{G}/H^{\perp}$, we only need to prove the direct implication. 
    
    Let $\psi$ be a KD-positive state. We denote by 
    \begin{equation}
        S=\mathrm{supp}(\psi):= \left\{g\in G, \psi(g) \neq 0\right\} \mathrm{\ and \ } S' = \mathrm{supp}(\hat{\psi}):= \left\{\chi\in \hat{G}, \hat{\psi}(\chi) \neq 0\right\}.
    \end{equation}
    Thus, for all $g\in S, \chi \in S'$, as $\psi$ is KD-positive, we have that:
    \begin{equation}\label{eq:KDpurepos1}
        \overline{\chi(g)}\psi(g)\overline{\hat{\psi}(\chi)} >0.
    \end{equation}
    We obtain that, for all $(g_1,g_2)\in S^2, \chi \in S'$:
    \begin{equation}\label{eq:KDpurepos2}
        \left(\overline{\chi(g_1)}\psi(g_1)\overline{\hat{\psi}(\chi)}\right)^{-1} >0 \mathrm{ \ and \ } \overline{\chi(g_2)}\psi(g_2)\overline{\hat{\psi}(\chi)} > 0.
    \end{equation}
    Thus, by multiplying the two inequalities in Eq.\eqref{eq:KDpurepos2}, for all $(g_1,g_2)\in S^2, \chi \in S'$
    \begin{equation}\label{eq:KDpurepos3}
        \chi(g_1-g_2)\psi(g_1)^{-1}\psi(g_2) > 0.
    \end{equation}
    Furthermore, by using Eq.\eqref{eq:KDpurepos3}, we obtain that, for all $(g_1,g_2)\in S^2, (\chi_1,\chi_2) \in S'^{2}$:
    \begin{equation}\label{eq:KDpurepos4}
        \chi_{1}(g_1-g_2)\overline{\chi_{2}(g_1-g_2)} > 0.
    \end{equation}
    As the action of the Heisenberg group preserves KD-positivity, we can suppose that $0\in S$ and $1\in S'$. We apply Eq.\eqref{eq:KDpurepos4} with $g_2=0$ and $\chi_2=1$ to obtain:
    \begin{equation}\label{eq:KDpurepos5}
        \forall (g,\chi) \in S\times S', \chi(g) > 0.
    \end{equation}
    As for all $(g,\chi) \in S\times S', \chi(g)\in\mcl{S}^{1}$, we have, with Eq.\eqref{eq:KDpurepos5}, that $ \chi(g)=1$. This implies that $S'\subseteq S^{\perp}$. Consequently, as $\left|S^{\perp}\right| \leqslant  \frac{\left|G\right|}{\left|S\right|}$, we obtain that 
    \begin{equation}\label{eq:KDpurepos6}
        \card{S}\card{S'} \leqslant \card{G}.
    \end{equation}
    Moreover, we have that:
    \begin{equation}
             \card{G}=\card{G}\sum_{(g,\chi)\in G\times\hat{G}} \left|Q[\psi](g,\chi)\right| = \card{G}^{\frac{1}{2}}\sum_{(g,\chi)\in G\times\hat{G}} \left|\psi(g)\hat{\psi}(\chi)\right| = \card{G}^{\frac{1}{2}}\sum_{g\in S} \left|\psi(g)\right| \sum_{\chi\in S'}\left|\hat{\psi}(\chi)\right|.
    \end{equation}
    We apply the Cauchy-Schwarz inequality twice to obtain that
    \begin{equation}
        \card{G}\leqslant \card{G}^{\frac{1}{2}}\left(\sum_{g\in S} \left|\psi(g)\right|^2\right)^{\frac{1}{2}}\card{S}^{\frac{1}{2}}\left(\sum_{\chi\in S'}\left|\hat{\psi}(\chi)\right|^2\right)^{\frac{1}{2}}\card{S'}^{\frac{1}{2}} = \card{G}^{\frac{1}{2}}\card{S}^{\frac{1}{2}}\card{S'}^{\frac{1}{2}} 
    \end{equation}
    Squaring this inequality and reorganising the terms, we have that:
    \begin{equation}\label{eq:KDpurepos7}
        \card{G} \leqslant \card{S}\card{S'}.
    \end{equation}
    Eq.\eqref{eq:KDpurepos6} and Eq.\eqref{eq:KDpurepos7} together imply that $\card{G} = \card{S}\card{S'}$ and thus, $\card{S'} = \card{S^{\perp}}$. Thus $S'=S^{\perp}$ and $S'$ is a subgroup of $\hat{G}$. Similarly, $S$ is a subgroup of $G$. And by equality case in the Cauchy-Schwarz inequality, we obtain that $\psi = \psi^{S}$, meaning that $\psi$ is in the Heisenberg orbit of a normalized characteristic function of a subgroup.  
\end{proof}

\section{Identification of Kirkwood-Dirac-real self-adjoint operators for finite abelian groups}
\label{sec:IdKDreal}
 In this section, we  characterise $\VR$, the set of self-adjoint operators that are KD-real:
\begin{Theorem}\label{thm:VKDr}
\begin{equation}\label{eq:VKDr}
\VR = \mathrm{span}_{\R}(\EcalKDCpu).
\end{equation}
\end{Theorem}
As a first step, using that $Q$ in Eq.\eqref{eq:Qmapdef} is bijective, we characterise in Lemma~\ref{lem:RealKD} the real functions on $G\times\widehat{G}$ for which the associated operator is self-adjoint.
A straightforward computation shows that the complex functions $\mathrm{Q}: G\times \hat{G} \to \C$, for which the associated operator $Q^{-1}[\mathrm{Q}]$ is self-adjoint, are those that satisfy
\begin{equation}\label{eq:condSAC}
     \forall (g,g')\in G^2, \quad \sum_{\chi\in \hat{G}} \chi(g-g')\left(\overline{\mathrm{Q}(g,\chi)} - \mathrm{Q}(g',\chi)\right)= 0.
\end{equation}
For real functions $\mathrm{Q}:G\times \hat{G} \to \R$, Eq.\eqref{eq:condSAC} becomes
\begin{equation}\label{eq:condSAR}
     \forall (g,g')\in G^2, \quad \sum_{\chi\in \hat{G}} \chi(g-g')\left(\mathrm{Q}(g,\chi) - \mathrm{Q}(g',\chi)\right)= 0.
\end{equation}
The following lemma characterises the functions satisfying Eq.\eqref{eq:condSAR}.

\begin{Lemma}\label{lem:RealKD}
    A real function $\mathrm{Q}: G\times \hat{G} \to \R$ satisfies Eq.\eqref{eq:condSAR} if and only if it belongs to $\sum_{H \mathrm{ \ subgroup \ of \ } G} P_{H}$ where $P_{H}$ is the set of real-valued functions on $G\times \hat{G}$ that are $H\times H^{\perp}$-periodic.
\end{Lemma}

Note that, for a subgroup $H$ of $G$, $(\eta^{H}_{g_0,\chi_0})_{(g_0,\chi_0)\in G/H\times\hat{G}/H^{\perp}}$ is an orthonormal basis of the set $P_{H}$ as, for any $f\in P_{H}$,
\begin{equation}\label{eq:PH1}
    f = \sum_{(g_0,\chi_0)\in G/H\times\hat{G}/H^{\perp}} \card{G}f(g_0,\chi_0)\eta^{H}_{g_0,\chi_0}.
\end{equation}

\begin{proof}
    As for any subgroup $H$ and any $(g_0,\chi_0)\in G/H\times\hat{G}/H^{\perp}, \eta^{H}_{g_0,\chi_0}$ satisfies Eq.\eqref{eq:condSAR}, we only need to prove the direct implication.
    
    Suppose that a real function $\mathrm{Q}:G\times \hat{G} \to \R$ satisfies Eq.\eqref{eq:condSAR}. We consider, for all $g\in G$, the map
    \begin{equation}
        \mathrm{Q}_g: \chi\in \hat{G} \mapsto \mathrm{Q}(g,\chi)\in\R.
    \end{equation}
    Note that, for all $(g,\chi)\in G\times\hat{G}$,
    \begin{equation}\label{eq:sumQ}
     \quad \mathrm{Q}(g,\chi) = \frac{1}{\card{G}^{\frac{1}{2}}}\sum_{g'\in G} \widecheck{\mathrm{Q}_{g}}(g')\overline{\chi(g')} = \frac{1}{2\card{G}^{\frac{1}{2}}}\sum_{g'\in G} \widecheck{\mathrm{Q}_{g}}(g')\overline{\chi(g')} + \widecheck{\mathrm{Q}_{g}}(-g')\overline{\chi(-g')}.
    \end{equation}
    For any real-valued function $\mathrm{Q}$, satisfying Eq.\eqref{eq:condSAR} is equivalent to:
    \begin{equation}\label{eq:EquivSAR}
        \forall (g,g')\in G^{2}, \widecheck{\mathrm{Q}_{g'}}(g-g') = \widecheck{\mathrm{Q}_{g}}(g-g') \mathrm{ \ and \ }  \widecheck{\mathrm{Q}_{g}}(-g') = \overline{\widecheck{\mathrm{Q}_{g}}(g')}.
    \end{equation}
    Note that the first equation in Eq.\eqref{eq:EquivSAR} implies that:
    \begin{equation}
        \forall (g,g')\in G^{2}, \widecheck{\mathrm{Q}_{g}}(g') = \widecheck{\mathrm{Q}_{g-g'}}(g').
    \end{equation} 
    Thus, for a fixed $h\in G$, we have that,
    \begin{equation}
        \forall g\in G,\forall n\in\N, \widecheck{\mathrm{Q}_{g}}(h) = \widecheck{\mathrm{Q}_{g+nh}}(h),
    \end{equation}
    meaning that the function $g\in G \to \widecheck Q_{g}(h)$ is $\langle h \rangle$-periodic where $\langle h \rangle$ is the subgroup of $G$ generated by $h$.
    It implies that the function $(g,\chi)\in G\times \hat{G} \to \widecheck{\mathrm{Q}_{g}}(h)\overline{\chi(h)} + \widecheck{\mathrm{Q}_{g}}(-h)\overline{\chi(-h)}$ is $H\times H^{\perp}$-periodic for $H=\langle h \rangle$ and real-valued. We thus obtain that any real-valued function $\mathrm{Q}$ satisfying Eq.\eqref{eq:condSAR} belongs to $\displaystyle\sum_{H \mathrm{ \ subgroup \ of \ } G} P_{H}$.
\end{proof}

We then have the following characterisation of self-adjoint operators with a real KD distribution.
\begin{Cor}\label{cor:SAO}
    Let $F$ be an operator such that Q[F] is real. The following two statements are equivalent: 
    \begin{itemize}[leftmargin=*,label=\textbullet]
        \item $F$ is a self-adjoint operator
        \item Q[F] belongs to $\displaystyle \sum_{H \mathrm{ \ subgroup \ of \ } G} P_{H}$.
    \end{itemize}
\end{Cor}

The proof of Theorem~\ref{thm:VKDr} now follows from the observation that $P_{H}$ is generated by $\left( \eta^{H}_{g_0,\chi_0}\right)_{(g_0,\chi_0)\in G/ H \times \hat{G}/ H^{\perp}}.$ 
This characterisation of the KD-real observables is a first step towards the characterisation of the KD-positive states, to which we turn next. 

\section{Characterisation of the Kirkwood-Dirac-positive states}\label{sec:CharKDpos}
 We give, in Theorem~\ref{thm:KDpositivestates}, a necessary and sufficient condition guaranteeing that $\EcalKDC = \conv{\EcalKDCpu}$. 
 
 First, we characterise the positive functions $\mathrm{Q}: G\times \hat{G} \to \R$ that are $H\times H^{\perp}$-periodic.
\begin{Lemma}\label{lem:posFunc}
Let $H\subseteq G$ be a subgroup. The space of positive $H\times H^{\perp}$-periodic functions $\mathrm{Q}: G\times \hat{G} \to \R$ coincides with $\mathrm{span}_{\R^{+}}\left(\left\{\eta^{H}_{g_0,\chi_0}\right\}_{(g_0,\chi_0)\in G/ H \times \hat{G}/ H^{\perp} }\right)$.
\end{Lemma}
The proof of this Lemma is a direct consequence of Eq.\eqref{eq:PH1}. 
Combining Lemma~\ref{lem:posFunc} with Corollary~\ref{cor:SAO}, we obtain a necessary and sufficient condition guaranteeing that all KD-positive states are convex combinations of pure KD-positive states.

\begin{Theorem}\label{thm:KDpositivestates}
    The following two statements are equivalent:
\begin{itemize}[leftmargin=*]
        \item for any positive $\mathrm{Q}:G\times \hat{G}\to \R$ such that $Q^{-1}[\mathrm{Q}]$ is a positive operator on $L^2(G)$, there exists $(\mathrm{Q}_{H})_{H \mathrm{ \ subgroup \ of \ } G}$ such that 
        \begin{itemize}
            \item $\mathrm{Q} = \sum_{H\mathrm{ \ subgroup \ of \ } G} \mathrm{Q}_{H}$
            \item $\mathrm{Q}_{H}$ is positive and $H\times H^{\perp}$-periodic;
        \end{itemize}
    \item $\EcalKDC = \conv{\EcalKDCpu}$.
\end{itemize}
\end{Theorem}
\begin{proof}
    We prove the direct implication. Let $\rho\in \EcalKDC$. Then, $Q[\rho]$ is a function that satisfies Eq.\eqref{cor:SAO} and that is positive as $\rho$ is KD-positive. Thus, by hypothesis,
    \begin{equation}\label{eq:qrhoqh}
        Q[\rho] = \sum_{H\mathrm{ \ subgroup \ of \ } G} \mathrm{Q}_{H},
    \end{equation}
    where each $\mathrm{Q}_H$ is positive and $H\times H^\perp$-periodic.
    It then follows from Lemma~\ref{lem:posFunc} and Eq.\eqref{eq:PH1} that
    \begin{equation}\label{eq:PH1bis}
    \mathrm{Q}_H = \sum_{(g_0,\chi_0)\in G/H\times\hat{G}/H^{\perp}} \card{G}\mathrm{Q}_H(g_0,\chi_0)\eta^{H}_{g_0,\chi_0}.
\end{equation}
Applying $Q^{-1}$ to both sides of this equation, and recalling that the  $\left(\eta^{H}_{g_0,\chi_0}\right)_{(g_0,\chi_0)\in G/H\times\hat{G}/H^{\perp}}$ are positive, one concludes that $Q^{-1}[\mathrm{Q}_H]$ is a positive and KD-positive operator. Hence
 \begin{equation}\label{eq:convex_1bis}
        \rho = \sum_{H\mathrm{ \ subgroup \ of \ } G} \sum_{(g_{0},\chi_{0})\in G/H\times G/H^{\perp}} |G|\mathrm{Q}_H(g_0, \chi_0)\cket{\psi^{H}_{g_0,\chi_0}}\bra{\psi^{H}_{g_0,\chi_0}}.
 \end{equation}
As $\mathrm{Tr}(\rho)=1$, Eq.\eqref{eq:convex_1bis} concludes the proof as $\rho$ is written as a convex combination of KD-positive pure states.
    
    We now prove the indirect implication. Suppose $\mathrm{Q}: G\times \hat{G}\to \R$ is a positive function for which $F = Q^{-1}[\mathrm{Q}]$ is a positive operator. If $F$ is the zero operator, there is nothing to prove. Suppose now that $F$ is not the zero operator, then $\frac{F}{\mathrm{Tr}(F)}$ is a KD-positive state as $Q[F]$ is positive. Thus, $\frac{F}{\mathrm{Tr}(F)}$ is a convex combination of KD-positive pure states \textit{i.e.} there exists $\left(\lambda_{g_{0},\chi_{0}}^{H}\right)\in\R^{+}$ such that
    \begin{equation}
    F = \mathrm{Tr}(F) \sum_{(g_{0},\chi_{0},H)} \lambda_{g_{0},\chi_{0}}^{H} \cket{\psi^{H}_{g_0,\chi_0}}\bra{\psi^{H}_{g_0,\chi_0}} \mathrm{ \ and \ }  \sum_{(g_{0},\chi_{0},H)} \lambda_{g_{0},\chi_{0}}^{H} = 1.
    \end{equation}
    Composing by $Q$, we have
    \begin{eqnarray}
        \nonumber \mathrm{Q} &=& \sum_{H} Q\left(\mathrm{Tr}(F) \sum_{(g_{0},\chi_{0})} \lambda_{g_{0},\chi_{0}}^{H} \cket{\psi^{H}_{g_0,\chi_0}}\bra{\psi^{H}_{g_0,\chi_0}}\right)\\
        &=&\sum_{H} \sum_{(g_{0},\chi_{0})} \mathrm{Tr}(F)\lambda_{g_{0},\chi_{0}}^{H} \eta^{H}_{g_0,\chi_0},
    \end{eqnarray}
    which shows, by Lemma~\ref{lem:posFunc}, that $\mathrm{Q} =\sum_{H} \mathrm{Q}_{H} $ where $\mathrm{Q}_{H}$ is positive and $H\times H^{\perp}$-periodic.
\end{proof}

In the next section, we apply this result to the discrete Fourier transform, in order to extend to prime power dimensions a Theorem in~\cite{langrenez2023characterizing2}, that states that $\EcalKDC = \conv{\EcalKDCpu}$ in prime dimensions.

\section{Discrete Fourier transform in prime power dimensions}\label{sec:DFtprimepower}
The goal of this section is to prove the following theorem.
\begin{Theorem}\label{thm:primepowers}
    Let $p$ be a prime number and $k\in \N^{*}$. Then, for $G=\Z_{p^{k}}$, we have that 
    \begin{equation}
     \EcalKDC=\conv{\EcalKDCpu}. 
    \end{equation}
\end{Theorem}

The proof relies on Lemma~\ref{lem:RealKD} and Theorem~\ref{thm:KDpositivestates}. We need in addition the following lemma that will allow us to exploit the particular structure of the subgroups of $\Z_d$, when $d=p^k$.
\begin{Lemma}\label{lem:Sum}
Let $G$ and $K$ be finite abelian groups. Suppose that there exist $N\in\N$ and two filtrations $G_{0}\subseteq G_{1} \subseteq \dots \subseteq G_{N} \subseteq G$ and $K_{N}\subseteq K_{N-1} \subseteq \dots \subseteq K_{1} \subseteq K_{0} \subseteq K$. If $f:G\times K \to \R$ is nonnegative and $f = \sum_{i=0}^{N} f_{i}$ where for all $i\in\IntEnt{0}{N}, f_{i}$ is $G_{i}\times K_{i}$-periodic, then there exists $(\tilde{f}_{i})_{i\in\IntEnt{0}{N}}$ such that $f = \sum_{i=0}^{N} \tilde{f}_{i}$ and for all $i\in\IntEnt{0}{N}, \tilde{f}_{i}$ is $G_{i}\times K_{i}$-periodic and nonnegative.
\end{Lemma}
\begin{proof}
    We prove this lemma by induction on $N$. 
    For $N=0$, $f=f_{0}$ is nonnegative and $\tilde{f}_{0}=f_{0}$ is nonnegative and $G_0\times K_0$-periodic.

    Suppose now that the lemma is true at rank $N-1$, $N\geq1$. Suppose that $G$ and $K$ are two groups equipped with filtrations $G_{0}\subseteq G_{1} \subseteq \dots \subseteq G_{N} \subseteq G$ and $K_{N}\subseteq K_{N-1} \subseteq \dots \subseteq K_{1} \subseteq K_{0} \subseteq K$. Let $f:G\times K \to \R$ be a nonnegative function such that $f = \sum_{i=0}^{N} f_{i}$ and $i\in\IntEnt{0}{N}, f_{i}$ is $G_{i}\times K_{i}$-periodic. For $i\in\IntEnt{1}{d}$,  as  $G_{1} \subseteq G_{i}, f_{i}$ is  $G_{1}\times K_{i}$-periodic. Thus, for all $(g,k,s) \in G\times K \times G_{1}$:
    \begin{equation}\label{eq:Rec1}
     \begin{array}{rcl}
          f(g+s,k) - f(g,k) &=&\displaystyle \sum_{i=0}^{N} f_{i}(g+s,k) - f_{i}(g,k) \\
          &=& \displaystyle f_{0}(g+s,k) - f_{0}(g,k) + \sum_{i=1}^{N} f_{i}(g+s,k) - f_{i}(g,k) \\
          &=& \displaystyle f_{0}(g+s,k) - f_{0}(g,k).
     \end{array}
    \end{equation}
    We want to define $\tilde{f}_{0}:G\times K \to \R $ so that $\tilde{f}_0$ is $G_0\times K_0$-periodic and that $0 \leqslant \tilde{f}_0 \leqslant f$. For that purpose, we define for all $(g,k)\in G\times K$,
    \[
    \tilde{f}_{0}(g,k) = f_{0}(g,k) - \min_{s\in G_{1}} f_{0}(g+s,k).
    \]
    As $f_0$ is $G_0\times K_0$-periodic, $\tilde{f}_{0}$ is also $G_0\times K_0$-periodic. We also have that for all $(g,k)\in G\times K$,
    \[
    0 \leqslant \tilde{f}_{0}(g,k) = f_{0}(g,k) - \min_{s\in G_{1}} f_{0}(g+s,k).
    \]
    Since $G_1$ is a finite group, there exists $s_1\in G_{1}$ such that $f_{0}(g+s_1,k) = \min_{s\in G_{1}} f_{0}(g+s,k)$ and thus, using Eq.\eqref{eq:Rec1}, for all $(g,k)\in G\times K$,
    \[
     0 \leqslant \tilde{f}_{0}(g,k) =  f_{0}(g,k) - f_{0}(g+s_{1},k) = f(g,k) - f(g+s_1,k) \leqslant f(g,k)  
    \]
    as $f$ is nonnegative.
    We define $\tilde{f} = f-\tilde{f}_{0}$. Then $\tilde{f}$ is nonnegative and 
    \begin{equation}\label{eq:tildef}
    \tilde{f} = f_0 - \tilde{f}_{0} + f_{1} + \sum_{i=2}^{N}f_{i}=\sum_{i=1}^N e_i
    \end{equation}
    where the $e_{i}$ are $G_{i}\times K_{i}$-periodic for $i\in\IntEnt{1}{N}$, since $f_0 - \tilde{f}_{0} + f_{1}$ is $G_{1}\times K_{1}$-periodic. 
    By the induction hypothesis, there then exist $(\tilde{e}_{i})_{i\in\IntEnt{1}{N}}$ such that $\tilde{f} = \sum_{i=1}^{N} \tilde{e}_{i}$ and for all $i\in\IntEnt{1}{N}$, $\tilde{e}_{i}$ is nonnegative and $G_{i}\times K_{i}$-periodic. This concludes the induction as $f = \tilde{f}_{0} + \sum_{i=1}^{N} \tilde{e}_{i}$ since for all $i\in\IntEnt{1}{N}$, $\tilde{e}_{i}$ is nonnegative and $G_{i}\times K_{i}$-periodic and $\tilde{f}_{0}$ is positive and $G_{0}\times K_{0}$-periodic.
\end{proof}

\noindent\textit{Proof of Theorem~\ref{thm:primepowers}.}

As $G=\Z_{p^{k}}$, a subgroup $H$ of $G$ is of the form $\Z_{p^{m}}$ for $m\in \IntEnt{0}{k}$. We will denote by $H_{m}$ the subgroup $\Z_{p^{m}}$, for $0\leq m\leq k$.
One has that
\begin{equation}
    H_0 \subset H_1 \subset \dots \subset H_k.
\end{equation} Thus, we also have that
\begin{equation}
    (H_k)^{\perp} \subset (H_{k-1})^{\perp} \subset \dots \subset (H_0)^{\perp}.
\end{equation}
To apply Theorem~\ref{thm:KDpositivestates}, we suppose that $\mathrm{Q}: G\times \hat{G}$ is a positive function that satisfies Eq.\eqref{eq:condSAR}. Lemma~\ref{lem:RealKD} then implies that there exist $(\mathrm{Q}_{m})_{m\in\IntEnt{0}{k}}$ such that:
\begin{itemize}
    \item $\mathrm{Q} = \sum_{m=0}^{k} \mathrm{Q}_{m}$,
    \item $\mathrm{Q}_m$ is $H_{m}\times H_{m}^{\perp}$-periodic.
\end{itemize}
Lemma~\ref{lem:Sum} implies that there exist $(\tilde{\mathrm{Q}}_{m})_{m\in\IntEnt{0}{k}}$ such that:
    \begin{itemize}
    \item $\mathrm{Q} = \sum_{m=0}^{k} \tilde{\mathrm{Q}}_{m}$,
    \item $\tilde{\mathrm{Q}}_m$ is non-negative and $H_{m}\times H_{m}^{\perp}$-periodic.
\end{itemize}
Using Theorem~\ref{thm:KDpositivestates}, we conclude that $\EcalKDC = \conv{\EcalKDCpu}$, which ends the proof. \qed

\section{Two examples for which $\conv{\EcalKDCpu} \subsetneq \EcalKDC$}\label{sec:ExamplesFT}

A natural question to ask is whether Theorem~\ref{thm:primepowers} extends to $G=\Z_d$ for composite dimensions $d$ or to other abelian groups $G$. In this section we show this not to be the case in general by analyzing the particular case where $G=\Z_6$, as well as $G=\Z_2\times\Z_2$, which arises in the study of two-qubit systems.

\subsection{Discrete Fourier transform in dimension 6: construction of a KD-positive state outside $\conv{\EcalKDCpu}$}
We focus on the particular example $G=\Z_{6}$. We will construct a state $\rho$ that is KD-positive and that is not in the convex hull of the set of KD-positive states.

All $24$ pure KD-positive states are explicitly given by Theorem~\ref{thm:KDPure} when $G= \Z_6$ and their KD distributions are readily determined. It is then possible to numerically compute the bounding planes of $Q[\conv{\EcalKDCpu}]$. We will focus on one particular bounding plane given by the following KD distribution
\[
Q_{\star} = \begin{pmatrix}
    10&10&1&10&-2&7\\
    10&10&7&-2&10&1\\
    7&1&-2&1&-5&-2\\
    -2&10&-5&-2&-2&1\\
    10&-2&1&-2&-2&-5\\
    1&7&-2&-5&1&-2
\end{pmatrix}.
\]
For any pure KD-positive state $\cket{\psi}$, one has:
\begin{equation}\label{eq:dim6KDpos}
    \mathrm{Tr}\left(Q_{\star}^{\dagger}Q[\psi]\right) \geqslant 0.
\end{equation}
We will construct a mixed KD-positive state $\rho$ such that
\[
\mathrm{Tr}\left(Q_{\star}^{\dagger}Q[\rho]\right) < 0.
\]
Consider $\alpha = \frac{1+\sqrt{3} + \sqrt{8+2\sqrt{3}}}{2}$ and the following KD distribution
\begin{equation}
Q_{\alpha} = \frac{1}{36\alpha+12}\begin{pmatrix}
1&0&2\alpha+1&1&\alpha&1 \\
1 & \alpha& 1 & 2\alpha+1 & 0 & 1 \\
0 & \alpha-1 & 2\alpha & \alpha & \alpha-1 & \alpha \\
2\alpha+1 & \alpha & 2\alpha+1 & 2\alpha+1 & 2\alpha & 1 \\
1 & 2\alpha & 2\alpha+1 & 2\alpha+1& \alpha & 2\alpha+1\\
\alpha & \alpha-1 & \alpha & 2\alpha & \alpha-1 & 0
\end{pmatrix}.
\end{equation}
We can readily compute
\begin{equation}\label{eq:dim6KDex}
\mathrm{Tr}\left(Q_{\star}^{\dagger}Q_\alpha\right) = \frac{3-3\alpha}{3\alpha+1} <0.
\end{equation}
We thus define $\rho_{\alpha} = Q^{-1}[Q_{\alpha}]$. Then $\rho_{\alpha}$ is a non-negative operator (its first five minors are strictly positive and the minor of order $6$ is 0, which implies that $\rho_{\alpha}$ has positive eigenvalues) with $\mathrm{Tr}(\rho_{\alpha}) = 1$. Thus, $\rho_{\alpha}$ is a KD-positive state, as $Q_{\alpha}$ only has positive entries, that is not in the convex hull of the set of KD-positive states according to Eq.\eqref{eq:dim6KDpos} and Eq.\eqref{eq:dim6KDex}.
Consequently, in this particular case, the inclusions in Eq~\eqref{eq:KDgeom} become
\begin{equation}
    \conv{\{\cket{a_g}\bra{a_g}, \cket{b_\chi}\bra{b_\chi}\mid g\in G, \chi\in \hat{G}\}}\subsetneq \EcalKDCpu\subsetneq \EcalKDC.
\end{equation}

\subsection{Fourier transform on $\Z_{2}\times \Z_{2}$}
We now focus on the example of $\Z_{2}\times \Z_{2}$ and will again exhibit a mixed KD-positive state that is not in $\conv{\EcalKDCpu}$. The construction is slightly different from the one used in the previous subsection.  

In this case, the transition matrix between the two bases is given by
\[
U = \frac{1}{2}\begin{pmatrix}
    1 & 1 & 1 & 1 \\
    1 & -1 & 1 & -1 \\
    1 & 1 & -1 & -1 \\
    1 & -1 & -1 & 1 \\
\end{pmatrix},
\]
which is the normalised real Hadamard matrix of dimension $4$. For simplicity of notation, the group $\Z_{2}\times \Z_{2}$ will be written as $\left\{ 00,01,10,11\right\}$. Following Theorem~\ref{thm:KDPure}, we can readily identify the $20$ pure KD-positive states. We will not list them here. Then, we can numerically determine the bounding planes of $\conv{\EcalKDCpu}$ inside $\VR$, as well as their respective normal vectors. We focus on one particular bounding plane, for which the normal vector in $\VR$ is given by
\[
V_{\star} = \frac{1}{20} \begin{pmatrix}
    1 & -4 & -4 & 8 \\
    -4 & 9 & 0 & 4 \\
    -4 & 0 & 9 & 4 \\
    8 & 4 & 4 & 1
\end{pmatrix}.
\]
This particular self-adjoint operator satisfies that
\begin{itemize}[leftmargin=*]
    \item[$\bullet$] for all pure KD-positive  states $\cket{\psi}$, $\bra{\psi}V_{\star}\cket{\psi}\leqslant 0.45$;
    \item[$\bullet$] $\mathrm{Tr}(V_{\star}^{\dagger} V_{\star}) = 1.05$;
    \item[$\bullet$] $\mathrm{Tr}(V_{\star}) = 1$.
\end{itemize}
Note however that $V_{\star}$ is not a positive operator and that $Q[V_{\star}]$ is real, but has both positive and negative entries. 
We can now construct a mixed KD-positive state $\rho_\lambda$ that satisfies $\mathrm{Tr}(\rho_{\lambda} V_{\star}) > 0.45$, as follows. We consider a full-rank state $\rho_{\star}$ that belongs to the bounding plane (so that it is KD-positive) and we mix it with~$V_{\star}$.
Specifically,
\begin{equation}
\begin{array}{rcl}
\rho_{\star}  &=& \displaystyle\frac{1}{2}\cket{a_{01}}\bra{a_{01}} + \frac{1}{8}\left(\frac{\cket{a_{00}} - \cket{a_{01}}}{\sqrt{2}}\right)\left(\frac{\bra{a_{00}} - \bra{a_{01}}}{\sqrt{2}}\right) + \frac{1}{8}\left(\frac{\cket{a_{10}} + \cket{a_{11}}}{\sqrt{2}}\right)\left(\frac{\bra{a_{10}}+ \bra{a_{11}}}{\sqrt{2}}\right) \\
&& \displaystyle +\frac{1}{8}\left(\frac{\cket{a_{00}} - \cket{a_{10}}}{\sqrt{2}}\right)\left(\frac{\bra{a_{00}} - \bra{a_{10}}}{\sqrt{2}}\right) + \frac{1}{8}\left(\frac{\cket{a_{01}} + \cket{a_{11}}}{\sqrt{2}}\right)\left(\frac{\bra{a_{01}} + \bra{a_{11}}}{\sqrt{2}}\right),
\end{array}
\end{equation}
and we consider the convex combination $ \rho_{\lambda} =  (1-\lambda)\rho_{\star}+\lambda V_{\star}$ for $\lambda\in [0,1]$. Then $\rho_{\lambda}$ is a self-adjoint operator and one checks numerically that $\rho_{\lambda}$ is a positive operator, and hence a state, if $\lambda\in [0,0.05]$. Moreover, a numerical computation shows that all entries of $Q[\rho_{\lambda}]$ are nonnegative if $\lambda\in[0,\frac{5}{19}]$.
Hence, provided $\lambda\in [0,0.05]$, $\rho_\lambda$ is a KD-positive state. On the other hand, as $\mathrm{Tr}(\rho_{\star} V_{\star}) = 0.45$, we obtain that
\[
\forall \lambda \in ]0,1], \quad \mathrm{Tr}(\rho_{\lambda}V_{\star}) = 0.45 + 0.6\lambda>0.45.
\]
Thus, if $\lambda\in ]0,0.05]$, $\rho_{\lambda}$ is a mixed KD-positive state that is not in $\conv{\EcalKDCpu}$.
This counterexample proves that, if $G = \Z_2\times\Z_2$, Eq.\eqref{eq:KDgeom} now reads:
\[
\conv{\{\cket{a_g}\bra{a_g}, \cket{b_\chi}\bra{b_\chi}\mid g\in G, \chi\in \hat{G}\}}\subsetneq \conv{\EcalKDCpu}\subsetneq \EcalKDC.
\]
\section{Conclusion and discussion}

We  have defined and studied the Kirkwood-Dirac representation (on $\Hcal=L^2(G)$)  associated to the  Fourier transform on arbitrary finite abelian groups $G$, generalizing the well known construction for the discrete Fourier transform, associated to the group $\Z_d$.  In this general context, we have identified the set $\EcalKDCpu$ of pure Kirkwood-Dirac-positive states as the Heisenberg orbit of the characteristic function of any subgroup of $G$. We have then shown that the real vector space of observables on $\Hcal=L^2(G)$ that are Kirkwood-Dirac real is equal to $\spanR{(\EcalKDCpu)}$. 

We have subsequently turned to the question whether $\conv{\EcalKDCpu} \subsetneq \EcalKDC$, or $\conv{\EcalKDCpu} = \EcalKDC$, where $\EcalKDC$ is the convex set of all Kirkwood-Dirac-positive states. We have provided a group-theoretic characterization of the situation where $\conv{\EcalKDCpu} = \EcalKDC$, which is the simplest one possible: the set of Kirkwood-Dirac-positive states is then a polytope with known vertices. We have used this criterium to show that, when $G=\Z_d$, with $d=p^k$ a prime power, one does indeed have that $\conv{\EcalKDCpu} = \EcalKDC$, so that all Kirkwood-Dirac-positive states are mixtures of the known pure Kirkwood-Dirac-positive states. This generalizes previously results for $k=1$~\cite{langrenez2023characterizing2} and $k=2$~\cite{Yang_2024}.

The question remains whether this simple geometry of $\EcalKDC$ occurs for other abelian groups. We provide two examples where it does not: $G=\Z_6$ and $G=\Z_2\times \Z_2$. In both these examples, we
exhibit Kirkwood-Dirac-positive states that are not mixtures of pure Kirkwood-Dirac-positive states.  What the situation is for arbitrary abelian groups remains an open question. 

\medskip
\noindent \textit{Acknowledgements: } The authors acknowledge the support of the CDP C2EMPI, as well as of the French State under the France-2030 programme, of the University of Lille, of the Initiative of Excellence of the University of Lille, of the European Metropolis of Lille for their funding and support of the R-CDP-24-004-C2EMPI project. This work was supported in part by the CNRS through the MITI interdisciplinary programs. 

D.R. acknowledges funding by the European Union (ERC, FourIntExP, 101078782). Views and opinions expressed are those of the author(s) only and do not necessarily reflect those of the European Union or European Research Council (ERC). Neither the European Union nor ERC can be held responsible for them.
\bibliographystyle{ieeetr}
\bibliography{Bibliography}

\begin{thebibliography}{10}

\bibitem{wigner1932}
E.~Wigner, ``On the {{Quantum Correction For Thermodynamic Equilibrium}},''
  {\em Physical Review}, vol.~40, no.~5, pp.~749--759, 1932.

\bibitem{Husimi40}
K.~Husimi, ``Some formal properties of the density matrix,'' {\em Proceedings
  of the Physico-Mathematical Society of Japan. 3rd Series}, vol.~22, no.~4,
  pp.~264--314, 1940.

\bibitem{Sudarshan63}
E.~C.~G. Sudarshan, ``Equivalence of semiclassical and quantum mechanical
  descriptions of statistical light beams,'' {\em Phys. Rev. Lett.}, vol.~10,
  pp.~277--279, 1963.

\bibitem{Glauber63}
R.~J. Glauber, ``Coherent and incoherent states of the radiation field,'' {\em
  Phys. Rev.}, vol.~131, pp.~2766--2788, 1963.

\bibitem{leonhardt}
U.~Leonhardt, {\em Essential Quantum Optics: From Quantum Measurements to Black
  Holes}.
\newblock Cambridge University Press, 2010.

\bibitem{Serafini}
A.~Serafini, {\em Quantum Continuous Variables: A Primer of Theoretical Methods
  (1st ed.)}.
\newblock Boca Raton: CRC Press., 2017.

\bibitem{hudson1974wigner}
R.~L. Hudson, ``When is the {W}igner quasi-probability density non-negative?,''
  {\em Reports on Mathematical Physics}, vol.~6, no.~2, pp.~249--252, 1974.

\bibitem{cohen1966}
L.~Cohen, ``Can {{Quantum Mechanics}} be {{Formulated}} as a {{Classical
  Probability Theory}}?,'' {\em Philosophy of Science}, vol.~33, no.~4,
  pp.~317--322, 1966.

\bibitem{cagl69a}
K.~Cahill and R.~J. Glauber, ``Ordered expansions in boson amplitude
  operators,'' {\em Physical Review}, vol.~177, no.~5, p.~1857, 1969.

\bibitem{cagl69b}
K.~Cahill and R.~J. Glauber, ``Density operators and quasi-probability
  distributions,'' {\em Physical Review}, vol.~177, no.~5, p.~1882, 1969.

\bibitem{Fe11}
C.~Ferrie, ``Quasi-probability representations of quantum theory with
  applications to quantum information science,'' {\em Reports On Progress in
  Physics}, vol.~74, p.~116001, 2011.

\bibitem{spekkens2008}
R.~W. Spekkens, ``Negativity and {{Contextuality}} are {{Equivalent Notions}}
  of {{Nonclassicality}},'' {\em Physical Review Letters}, vol.~101, no.~2,
  p.~020401, 2008.

\bibitem{arvidssonshukur2024properties}
D.~R.~M. Arvidsson-Shukur, W.~F. Braasch~Jr, S.~De~Bièvre, J.~Dressel, A.~N.
  Jordan, C.~Langrenez, M.~Lostaglio, J.~S. Lundeen, and N.~Yunger~Halpern,
  ``Properties and applications of the {Kirkwood–Dirac} distribution,'' {\em
  New Journal of Physics}, vol.~26, p.~121201, 2024.

\bibitem{lostaglio2023kirkwood}
M.~Lostaglio, A.~Belenchia, A.~Levy, S.~Hern{\'a}ndez-G{\'o}mez, N.~Fabbri, and
  S.~Gherardini, ``{K}irkwood-{D}irac quasiprobability approach to the
  statistics of incompatible observables,'' {\em Quantum}, vol.~7, p.~1128,
  2023.

\bibitem{arvidsson-shukur2020}
D.~R.~M. Arvidsson-Shukur, N.~Yunger~Halpern, H.~V. Lepage, A.~A. Lasek,
  C.~H.~W. Barnes, and S.~Lloyd, ``Quantum advantage in postselected
  metrology,'' {\em Nat. Commun.}, vol.~11, no.~1, p.~3775, 2020.

\bibitem{pashayan2015}
H.~Pashayan, J.~J. Wallman, and S.~D. Bartlett, ``Estimating outcome
  probabilities of quantum circuits using quasiprobabilities,'' {\em Phys. Rev.
  Lett.}, vol.~115, p.~070501, 2015.

\bibitem{gross2006}
D.~Gross, ``Hudson's theorem for finite-dimensional quantum systems,'' {\em
  Journal of Mathematical Physics}, vol.~47, no.~12, p.~122107, 2006.

\bibitem{vanherstraetencerf2021}
Z.~Van~Herstraeten and N.~J. Cerf, ``Quantum {{Wigner}} entropy,'' {\em
  Physical Review A}, vol.~104, no.~4, p.~042211, 2021.

\bibitem{herstraetenetal2023}
Z.~Van~Herstraeten, M.~G. Jabbour, and N.~J. Cerf, ``Continuous majorization in
  quantum phase space,'' {\em Quantum}, vol.~7, p.~1021, 2023.

\bibitem{diasprata2023a}
N.~C. Dias and J.~N. Prata, ``{On a Recent Conjecture by Z. Van Herstraeten and
  N. J. Cerf for the Quantum Wigner Entropy},'' {\em Annales Henri Poincaré},
  vol.~24, no.~7, pp.~2341--2362, 2023.

\bibitem{Bievre19}
S.~De~Bi{\`e}vre, D.~B. Horoshko, G.~Patera, and M.~I. Kolobov, ``Measuring
  nonclassicality of bosonic field quantum states via operator ordering
  sensitivity,'' {\em Physical Review Letters}, vol.~122, no.~8, p.~080402,
  2019.

\bibitem{ryspagal15}
S.~Ryl, J.~Sperling, E.~Agudelo, M.~Mraz, S.~Köhnke, B.~Hage, and W.~Vogel,
  ``Unified nonclassicality criteria,'' {\em Physical Review A: Atomic,
  Molecular, and Optical Physics}, vol.~92, no.~1, p.~011801, 2015.

\bibitem{arvidsson-shukuretal2021}
D.~R.~M. Arvidsson-Shukur, J.~C. Drori, and N.~Yunger~Halpern, ``Conditions
  tighter than noncommutation needed for nonclassicality,'' {\em Journal of
  Physics A: Mathematical and Theoretical}, vol.~54, no.~28, p.~284001, 2021.

\bibitem{debievre2021}
S.~De~Bièvre, ``Complete {Incompatibility}, {Support} {Uncertainty}, and
  {Kirkwood}-{Dirac} {Nonclassicality},'' {\em Physical Review Letters},
  vol.~127, no.~19, p.~190404, 2021.

\bibitem{DeB23}
S.~De~Bi\`evre, ``Relating incompatibility, noncommutativity, uncertainty, and
  {K}irkwood--{D}irac nonclassicality,'' {\em J. Math. Phys.}, vol.~64, no.~2,
  pp.~Paper No. 022202, 31, 2023.

\bibitem{XU2024}
J.~Xu, ``{K}irkwood-{D}irac classical pure states,'' {\em Physics Letters A},
  vol.~510, p.~129529, 2024.

\bibitem{langrenezetal2024}
C.~Langrenez, W.~Salmon, S.~De~Bièvre, J.~J. Thio, C.~K. Long, and D.~R.~M.
  Arvidsson-Shukur, ``The set of {Kirkwood-Dirac} positive states is almost
  always minimal.'' arXiv.2405.17557, 2024.

\bibitem{langrenez2023characterizing2}
C.~Langrenez, D.~R.~M. Arvidsson-Shukur, and S.~De~Bièvre, ``{Characterizing
  the geometry of the {K}irkwood–{D}irac-positive states},'' {\em Journal of
  Mathematical Physics}, vol.~65, no.~7, p.~072201, 2024.

\bibitem{langrenezetal2024a}
C.~Langrenez, S.~De~Bièvre, and D.~R.~M. Arvidsson-Shukur, ``Convex roofs
  witnessing {K}irkwood-{D}irac nonpositivity.'' arXiv.2407.04558, 2024.

\bibitem{MaOzPr04}
E.~Matusiak, M.~Özaydın, and T.~Przebinda, ``{The Donoho-Stark uncertainty
  principle for a finite abelian group},'' {\em Acta Mathematica Universitatis
  Comenianae}, vol.~73, no.~2, pp.~155--160.

\bibitem{Yang_2024}
Y.-H. Yang, S.~Yao, S.-J. Geng, X.-L. Wang, and P.-Y. Chen, ``{}geometry of
  kirkwood–dirac classical states: a case study based on discrete fourier
  transform,'' vol.~57, no.~43, p.~435303, 2024.

\bibitem{xu2024a}
J.~Xu, ``{Hermitian Kirkwood-Dirac real operators for discrete Fourier
  transformations}.''
\newblock {arXiv}:2412.16945.

\bibitem{thio2024}
J.~J. Thio, W.~Salmon, C.~H.~W. Barnes, S.~De~Bièvre, and D.~R.~M.
  Arvidsson-Shukur, ``Contextuality can be verified with noncontextual
  experiments.'' arXiv:2412.00199, 2024.

\end{thebibliography}
\end{document}